 \documentclass[letterpaper, 10 pt, conference]{ieeeconf}

\IEEEoverridecommandlockouts
\usepackage[latin1]{inputenc}
\usepackage[english]{babel}
\usepackage{amsmath,amssymb}
\usepackage{graphicx,graphicx,verbatim}
\usepackage[hypertex]{hyperref}
\usepackage{verbatim}
\usepackage{color}
\newtheorem{thm}{Theorem}
\newtheorem{cor}{Corollary}
\newtheorem{lem}{Lemma} 
\newtheorem{assumption}{Assumption}

\newtheorem{defi}{Definition}
\newtheorem{prop}{Proposition}

\newtheorem{exx}{Example}
\newtheorem{remm}{Remark}

\newenvironment{definition}{\begin{defi}\rm }{\hfill \hspace*{1pt} \hfill $\lrcorner$\end{defi}}
\newenvironment{remark}{\begin{remm}\rm }{\hfill \hspace*{1pt} \hfill $\lrcorner$\end{remm}}

\newenvironment{theorem}{\begin{thm} \rm }{\hfill \hspace*{1pt} \hfill $\lrcorner$\end{thm}}

\newenvironment{example}{\begin{exx}\rm }{\hfill \hspace*{1pt} \hfill $\lrcorner$ \end{exx}}
\newenvironment{proofof}{{\em Proof of }}{\hfill \hspace*{1pt}
\hfill $\blacksquare$}

\newcommand\real{\ensuremath{{\mathbb R}}}

\newcommand\mymatrix[2]{\left[\begin{array}{#1} #2 \end{array}\right]}
\newcommand{\smallmat}[1]{\left[ \begin{smallmatrix}#1
    \end{smallmatrix} \right]}

\newcommand{\calM}{\mathcal{X}}
\newcommand{\calU}{\mathcal{U}}
\newcommand{\calY}{\mathcal{Y}}

\begin{document}

\title{\LARGE \bf On differential passivity of physical systems}

\author{F. Forni, R. Sepulchre, A.J. van der Schaft
\thanks{This paper presents research results of the Belgian Network DYSCO
(Dynamical Systems, Control, and Optimization), funded by the
Interuniversity Attraction Poles Programme, initiated by the Belgian
State, Science Policy Office. The scientific responsibility rests with
its authors.}%
\thanks{F. Forni is with the Department of Electrical Engineering and Computer Science, 
University of Li{\`e}ge, 4000 Li{\`e}ge, Belgium, \texttt{fforni@ulg.ac.be}. 
His research is supported by FNRS.
R. Sepulchre is with the University of Cambridge, Department of Engineering, Trumpington Street, Cambridge CB2 1PZ, and with the Department of Electrical Engineering and Computer Science, 
University of Li{\`e}ge, 4000 Li{\`e}ge, Belgium, \texttt{r.sepulchre@eng.cam.ac.uk}.
A.J. van der Schaft is with the Johann Bernoulli Institute for Mathematics and
Computer Science, University of Groningen, 9700 AK, the Netherlands \texttt{a.j.van.der.schaft@rug.nl}.} }

\date{\today}

\maketitle

\begin{abstract}  
Differential passivity is a property that allows to check with a
pointwise criterion that a system is incrementally passive,
a property that is relevant to study interconnected systems in the context
of regulation, synchronization, and estimation.
The paper investigates how restrictive is the property, focusing
on a class of open gradient systems encountered in the coenergy
modeling framework of physical systems, in particular the Brayton-Moser
formalism for nonlinear electrical circuits.
\end{abstract}

\section{Introduction}

Motivated by the differential Lyapunov framework presented in \cite{Forni2012arxiv} to
study incremental stability, the recent papers \cite{Schaft2013}
and \cite{Forni2013} introduced the notion
of differential dissipativity to study incremental dissipativity, the analog
of incremental stability  for open systems. A related notion of tranverse
incremental dissipativity is presented in \cite{Manchester2012} 
to study limit cycles.  The interest
for incremental notions of stability and dissipativity stems from
analysis and design  problems concerned with a distance between arbitrary
solutions rather than a distance to a particular (equilibrium) solution : such
problems include regulation and tracking, estimation and observer design,
or synchronization, coordination, and entrainment.

The differential approach to study incremental properties is rooted in contraction
theory, following the influential paper of \cite{Lohmiller1998} 
in control theory. In short, incremental
properties of dynamical systems can be studied differentially,  through the
variational equations. The analysis of the variational equation (or more precisely
of the prolonged system) is appealing because it leads to pointwise
conditions to be verified on the prolonged vector field rather than on the solutions,
in the spirit of Lyapunov theory. The approach is geometric and the differential properties
are  potentially simpler to verify than their incremental counterparts.

The present paper pursues the developments of \cite{Schaft2013}
and \cite{Forni2013} to investigate
how restrictive it is to check differential passivity on a given system. More fundamentally,
we are interested in which class of physical systems are differentially passive and
what is the physical interpretation of the property, if any. The success of passivity as
an analysis and design concept of system theory stems from its clear energy interpretation
in physical systems: passivity expresses that the increase of internally stored energy
cannot exceed the energy supplied by the environment. It is still unclear whether a
similar interpretation exists for differential passivity.

We provide geometric conditions that characterize differential passivity with respect
to a quadratic storage and we further investigate the general conditions for a class
of gradient systems. Our motivation stems from the fact that a broad class of physical
models admits a gradient representation in the coenergy framework, see e.g. 
\cite{Jeltsema2009,Schaft2011},
after the work of Brayton and Moser for nonlinear electrical circuits.

The paper provides a number of simple examples that illustrate that differential
passivity may hold for a sizable class of physical models and that feedback
can help achieving the property, as for passivity.

The paper is organized as follows:
we revisit the notion of differential passivity in Section \ref{sec:passivity}, providing the
definitions of prolonged and variational system, differential storage, and differential
supply rate. Geometric conditions for passivity are summarized in Section \ref{sec:geometric}.
Section \ref{sec:gradient} studies the differential passivity
of gradient systems. Differential passivity for Brayton-Moser systems
is characterized in 
Section \ref{sec:brayton_moser}. 

\vspace{2mm}

\begin{small}
\noindent  {\bf Notation}:  
Given a manifold $\mathcal{X}$, and a point $x$ of $\mathcal{X}$, 
$T_x\mathcal{X}$ denotes the \emph{tangent space} of $\mathcal{X}$ at $x$.
$T\mathcal{X}:= \bigcup_{x\in\mathcal{X}} \{x\}\times T_x\mathcal{X}$
is the \emph{tangent bundle}.
Given two manifolds $\mathcal{X}_1$ and $\mathcal{X}_2$ and a mapping 
$f:\calM_1 \to \calM_2$, $f$ is
of class $C^k$, $k\in\mathbb{N}$, if its coordinate representation
is a $C^k$ function. 
A \emph{curve} $\gamma$ on a given manifold $\mathcal{X}$ is a mapping 
$\gamma :I \subset \real \to \calM$.
We sometime use
$\dot{\gamma}(t)$ to denote $\frac{\partial \gamma(t)}{\partial t}$.

$I_n$ is the identity matrix of dimension $n$.
Given a vector $v$, $v^T$ denotes the transpose vector 
of $v$. Given a matrix $M$ we say that 
$M\geq 0$ or $M\leq 0$  if $v^T M v \geq 0 $ or $v^T M v \leq 0$,
for each $v$, respectively.
Given the vectors $\{v_1,\dots,v_n\}$, 
$\mathrm{Span}(\{v_1,\dots,v_n\}) := \{v\,|\,\exists \lambda_1,\dots\lambda_n\in\real \mbox{ s.t. } v = \sum_{i=1}^n \lambda_i v_i\}$. In coordinates, we denote the differential of a function $f$
at $x$ by $\frac{\partial f(x)}{\partial x}$.
The Hessian of $f$ at $x$ is denoted by  $\frac{\partial^2 f(x)}{\partial x^2}$.

A \emph{distance} $d:\calM\times\calM\to \real_{\geq 0}$ 
on a manifold $\calM$ is a positive function that satisfies
$d(x,y) = 0$ if and only if $x=y$, for each $x,y\in \calM$ and 
$d(x,z) \leq d(x,y) + d(y,z)$ for each $x,y,z\in \calM$. 
A set $\mathcal{S}\subset\mathcal{X}$ is bounded if $\sup_{x,y\in\mathcal{S}} d(x,y) < \infty$ for any given
distance $d$ on $\mathcal{X}$. A curve $\gamma:I\to\calM$ is \emph{bounded} when its image is bounded.
Given a manifold $\mathcal{X}$, a set of \emph{isolated points} $\Omega\subset \mathcal{X}$
satisfies: for any distance function $d$ on $\mathcal{X}$ and any given pair 
$x_1,x_2$ in $\Omega$, there exists an $\varepsilon>0$ such that $d(x_1,x_2)\geq \varepsilon$.
\end{small}

\section{Differential passivity}
\label{sec:passivity}

\subsection{Prolonged systems}
Consider the nonlinear system $\Sigma$ with state space
$\mathcal{X}$, and inputs and outputs spaces $\calU\subset \real^m$ and $\calY\subset\real^m$, respectively,
given by 
\begin{equation}
\label{eq:sys}
\left\{
\begin{array}{rcl}
 \dot{x} & = & f(x) + g(x)u \\
 y &=& h(x)
\end{array}
\right.
\end{equation}
where $x \in \mathcal{X}$, and $u \in \calU$, and $y\in \calY$.
$f$ and $g_i$, $i\in\{1\dots m\}$ are vector fields. $h:\mathcal{X} \to \mathcal{Y}$.

Contraction analysis requires sufficient differentiability ($C^2$) 
of the solutions $\psi(t,x_0)$ 
to \eqref{eq:sys}, from 
any initial condition $x_0\in\mathcal{X}$ (see, e.g. \cite{Lohmiller1998,Russo2010}).
To enforce the desired regularity, 
we make the following standing assumption.
\begin{assumption}
\label{assume:regularity}
$f$ and $g_i$, $i\in\{1\dots m\}$, are $C^2$ vector fields ($g_i$ denotes the $i$-th column of $g$).
$h:\mathcal{X} \to \mathcal{Y}$ is a $C^2$ function. 
The input signal $u:\real \to \mathcal{U}$ is a $C^2$ function.
\end{assumption}

To a system of the form \eqref{eq:sys} one can associate the 
\emph{variational system} given by 
\begin{equation}
\label{eq:dsys}
\left\{
\begin{array}{rcl}
 \dot{\delta x} & = & \frac{\partial f(x)}{\partial x} \delta x 
 + \frac{\partial g(x)u}{\partial x} \delta x  + g(x) \delta u\\
 \delta y &=& \frac{\partial h(x)}{\partial x} \delta x \ .
\end{array}
\right.
\end{equation}
We call \emph{prolonged system} the combination of \eqref{eq:sys} and \eqref{eq:dsys}, following
\cite{crouch1987,Schaft2013}. A coordinate free representation of the prolonged
system is provided by the notions of complete and vertical lifts, 
as shown in \cite{crouch1987,Schaft2013}. 

Under Assumption \ref{assume:regularity}, 
for every solution $(x,u,y)(\cdot)$ to \eqref{eq:sys}, 
the solutions $(\delta x,\delta u,\delta y)(\cdot)$ to \eqref{eq:dsys}
represent infinitesimal variations on $(x,u,y)(\cdot)$, that is, 
the infinitesimal mismatch between $(x,u,y)(\cdot)$ and neighboring solutions.
This intuitive representation is clarified in Remark \ref{rem:displacement}. 
Pursuing this intuition, if the dynamics of \eqref{eq:dsys} guarantee that 
$\delta x$ converges to zero then,
necessarily, the solutions to \eqref{eq:sys} must converge towards each other. 
A Lyapunov-based analysis of the connection between contraction of $\delta x$ and
incremental stability can be found in \cite{Forni2012arxiv}. 
\begin{remark}
\label{rem:displacement}
For each $s\in[0,1]$ let $\gamma(s)$ be an initial condition for \eqref{eq:sys}
and $u(\cdot,s)$ an input signal. Assume that $\gamma(\cdot) \in C^2$ and 
$u(\cdot,\cdot) \in C^2$. Then, for each $s\in[0,1]$  
$x(\cdot,s)$ is a solution to \eqref{eq:sys} from the initial condition
$\gamma(s)$ under the action of the input $u(\cdot,s)$.
Define the displacement $\delta x(t,s) := \frac{\partial }{\partial s}x(t,s)$
and $\delta u(t,s) := \frac{\partial }{\partial s}u(t,s)$.
Then, by chain rule and differentiability, we have that
$\frac{d}{dt}\delta x(t,s) 
= \frac{\partial^2 }{\partial s\partial t} x(t,s)
= \frac{\partial}{\partial s} [ f(x(t,s)) + g(x(t,s))u(t,s) ]
=
\frac{\partial f(x(t,s))}{\partial x} \delta x(t,s) 
 + \frac{\partial g(x(t,s))u(t,s)}{\partial x} \delta x(t,s) 
 + g(x(t,s)) \delta u(t,s)$. Thus, $\delta x(\cdot,s)$ is a solution
 to \eqref{eq:dsys} from the initial condition $\frac{\partial \gamma(s)}{\partial s}$
 under the action of the input $\delta u(\cdot,s)$. Moreover, the output signal
 $\delta y(t,s)$ is given by  $\frac{\partial y(t,s)}{\partial s} = \frac{\partial h(x(t,s))}{\partial x}\delta x(t,s)$.
 \end{remark}

\subsection{Differential passivity}

Henceforth we provide the notion of differential storage
function and differential passivity. These notions
are taken from \cite[Sections 3 and 4]{Forni2013} and restrict
the definitions in \cite[Section 4]{Schaft2013} to the
case in which the function $P$ in \cite[Definition 4.1 and Proposition 4.3]{Schaft2013}
is a candidate Finsler-Lyapunov function \cite{Forni2012arxiv}.
This restriction makes possible the connection between differential passivity and
incremental stability.

\begin{definition}
 Let $\Omega$ be a set of isolated point in $\mathcal{X}$. 
 For each $x\in\mathcal{X}$, suppose that 
 $T_x\mathcal{X}$ can be subdivided into a \emph{vertical distribution}
$\mathcal{V}_x\subset T_x\mathcal{X}$ 
 \begin{equation}
  \label{eq:V_x}
  \mathcal{V}_x := \mathrm{Span}(\{v_1(x),\dots, v_r(x)\}) \qquad 0\leq r<d \ ,
 \end{equation}
 and a \emph{horizontal distribution} $\mathcal{H}_x\subseteq T_x\mathcal{X}$ 
 complementary to $\mathcal{V}_x$, i.e.
 $\mathcal{V}_x \oplus \mathcal{H}_x = T_x\mathcal{X}$, 
 \begin{equation}
 \label{eq:H_x}
 \mathcal{H}_x := \mathrm{Span}(\{h_1(x),\dots, h_q(x)\}) \qquad 0 < q \leq d-r \,
 \end{equation}
  where $v_i$, $i \in \{1,\dots,r\}$, 
 and $h_i$, $i\in\{1,\dots,q\}$, are $C^1$ vector fields.
 
 A function $\delta S:{T\mathcal{X}} \to \real_{\geq 0}$ is a \emph{differential storage function}
 for the dynamical system $\Sigma$ in \eqref{eq:sys} 
 if there exist $c_1,c_2\in \real_{\geq 0}$, $p\in\real_{\geq 1}$, 
 and a function $F:T\mathcal{X} \to \real_{\geq 0}$ such that,
 for each $(x,\delta x) \in T\mathcal{X}$,
  \begin{equation}
 \label{eq:FinsLyap_bounds}
  c_1\, F(x,\delta x)^p \ \leq \ \delta S(x,\delta x) \ \leq \ c_2\, F(x,\delta x)^p \ .
  \qquad 
 \end{equation}
 {
$\delta S$ and $F$ must satisfy the following conditions.
} 
 Given a set of isolated points $\Omega\subset\mathcal{X}$,
 \begin{itemize}
  \item[(i$_a$)] 
  $\delta S$ and $F$ are $C^1$, \quad $\forall x\in\mathcal{X}$, 
  $\forall \delta x\in \mathcal{H}_x\setminus\{0\}$; \vspace{1mm}
  \item[(i$_b$)]    
  $\delta S(x,\delta x) = \delta S(x,\delta x_h)$ 
  and $F(x,\delta x) = F(x,\delta x_h)$, \quad
  $\forall (x,\delta x)\in T\mathcal{X}$ such that 
  $(x,\delta x) = (x,\delta x_h) + (x,\delta x_v)$, 
  $\delta x_h \in \mathcal{H}_x$, and $\delta x_v \in \mathcal{V}_x$;\vspace{1mm}
  \item[(ii)] $F(x,\delta x) >0$, \quad 
  $\forall x\in\mathcal{X}\setminus\Omega$ 
  $\forall \delta x\in \mathcal{H}_x\setminus\{0\}$; \vspace{1mm}
  \item[(iii)] $F(x,\lambda \delta x) = \lambda F(x,\delta x)$, \quad   
   $\forall \lambda\!>\!0$, $\forall x\!\in\!\mathcal{X}$, 
   $\forall \delta x\!\in\! \mathcal{H}_x$; \vspace{1mm}
  \item[(iv)] $F(x,\delta x_1+\delta x_2) < F(x,\delta x_1) + F(x,\delta x_2)$, \quad \\
  $\forall x\in \mathcal{X}\setminus\Omega$ and 
  $\forall \delta x_1,\delta x_2\in \mathcal{H}_x\setminus\{0\}$ 
  such that $\delta x_1\neq \lambda \delta x_2$ for any given $\lambda \in \real$.
 \end{itemize} \vspace{-5mm}
 \end{definition} \vspace{2mm}

When $\mathcal{V}_x=\emptyset$, $F(x,\delta x)$ provides a
non symmetric norm on each tangent space $T_x\mathcal{X}$. 
A suggestive notation for $F$ is given by  $|\delta x|_x$ which
combined to \eqref{eq:FinsLyap_bounds} provides an intuitive
interpretation of the differential storage function $\delta S$
as a local measure of the displacement length. 
For $\mathcal{V}_x \neq \emptyset$, { it may occur that 
$\delta S(x,\delta x_1)= \delta S(x,\delta x_2)$ for
$0 \neq \delta x_1 - \delta x_2 \in \mathcal{V}_x$.
In such a case,}
$\delta S$ measures
the length of each $\delta x$ by looking only at its 
horizontal component. An example of a differential storage with 
$\mathcal{V}_x \neq 0$ is provided by 
$\delta S(x,\delta x) = \delta y^T \delta y$.

It is worth to mention that 
a differential storage function $\delta S$ is also a
horizontal Finsler-Lyapunov function \cite[Section VIII]{Forni2012arxiv}.
Therefore, $\delta S$ endows $\mathcal{X}$ with the structure of a pseudo-metric space,
connecting differential passivity and incremental stability \cite{Tamassy08,Shen00}.
An extended discussion and examples are provided in 
\cite[Sections IV and VIII]{Forni2012arxiv}. 
 
The notion of differential passivity introduced below 
is just passivity lifted to the tangent bundle. 
\begin{definition}
\label{def:dds}
The dynamical system $\Sigma$ in \eqref{eq:sys} is 
\emph{differentially passive} 
if there exists a differential storage function $\delta S$ such that 
\begin{equation}
\label{eq:passivity}
\delta S( x(t),\delta x(t)) - \delta S( x(0),\delta x(0)) 
\leq \int_0^t \delta y(\tau)^T\delta u(\tau) \  d\tau 
\end{equation}
for all $t \geq 0$ and all solutions
$(x,u,y,\delta x, \delta u, \delta y)(\cdot)$ to the prolonged
system \eqref{eq:sys},\eqref{eq:dsys}.
\end{definition}
The equivalent formulation 
$\frac{d}{dt} \delta S(x(t),\delta x(t)) \leq \delta y(t)^T\delta u(t)$
coincides with \cite[Definition 4.1]{Schaft2013}. In comparison
to passivity, differential passivity
builds a relation between the energy - or cost - $\delta S$ 
associated to an infinitesimal variation of the solution $x(t)$, and
the energy associated to an infinitesimal variation on the input/output signals. 
In comparison to incremental passivity \cite{Desoer1975,Stan2007}, 
$\delta y^T\delta u$ does not impose any prescribed form $\Delta y^T\Delta u = (y_1-y_2)^T(u_1-u_2)$
to the input/output mismatch. 
Instead, following Remark \ref{rem:displacement}, given a parameterization 
$u(s), y(s)$ such that $(u(0),y(0)) = (u_1,y_1)(\cdot)$ and
$(u(1),y(1)) = (u_2,y_2)(\cdot)$
we have that that $(y_1-y_2)^T(u_1-u_2)$ is replaced by 
$\int_0^1 \frac{\partial y(s)}{\partial s}^T \frac{\partial u(s)}{\partial s} ds$.
Note that 
$\int_0^1 \frac{\partial y(s)}{\partial s}^T \frac{\partial u(s)}{\partial s} \ ds
= \Delta y^T\Delta u$ only if $y(s) = sy_1 + y_2(1-s)$ and $u(s) = su_1 + u_2(1-s)$.
This is particularly relevant at integration along solutions, since 
an initial parameterization satisfying the identity above at time $t=0$ 
does not preserve the identity for $t>0$, in general (on nonlinear models). 
 
We conclude the section by illustrating two basic results 
of differential passivity. The reader is referred to 
\cite{Forni2013,Schaft2013} for further results.
\begin{theorem}
\label{thm:stability}
Let $\Sigma$ in \eqref{eq:sys} be differentially passive with a differential storage
$\delta S$ whose vertical distribution $\mathcal{V}_x = 0$ for each $x\in\mathcal{X}$. 
Then, \eqref{eq:sys} is incrementally stable.
\end{theorem}
\begin{proof}
For $\delta u = 0$, differential passivity guarantees that $\dot{\delta S} \leq 0$. 
For $\mathcal{V}_x = 0$,
$\delta S$ is  a Finsler-Lyapunov function, thus incremental stability
follows from \cite[Theorem 1]{Forni2012arxiv}.
\end{proof}
\begin{theorem}
\label{thm:interconnection}
Let $\Sigma_1$ and $\Sigma_2$ be differentially passive dynamical
systems \eqref{eq:sys}. Let $(u_i,y_i)$ be the input and the output of $\Sigma_i$, for $i=1,2$.
Then, the dynamical system $\Sigma$ arising from the feedback interconnection
\begin{equation}
\label{eq:output_feedback_interconnection}
u_1 = -y_2 + v_1 \ , \; u_2 = y_1 + v_2, 
\end{equation}
is differentially passive from $v=(v_1,v_2)\in \calU_1\times \calU_2$ 
to $y=(y_1,y_2)\in\calY_1\times\calY_2$.
\end{theorem}
\begin{proof}
Take $\delta S \!=\! \delta S_1 \!+\! \delta S_2$. 
$\dot{\delta S} \!\leq\! \delta y_1\delta v_1 \!+\! \delta y_2\delta v_2$. 
\end{proof}

\section{The geometry of differential passivity}
\label{sec:geometric}
For quadratic differential storage functions 
$\delta S = \frac{1}{2}\delta x^T M(x) \delta x$ (Riemannian metrics), $M(x)>0$,
the differential passivity
of systems of the form \eqref{eq:sys} is characterized geometrically
by the following conditions.
For each $x\in \mathcal{X}$ and $u\in\mathcal{U}$,
\begin{equation}
\label{eq:gen_cond1}
 \!\!M(x) \frac{\partial f(x)}{\partial x} + \frac{\partial f(x)}{\partial x}^T\!\! M(x) 
 + \sum_i \frac{\partial M(x)}{\partial x_i}[f(x))]_i  \leq 0 
\end{equation}
\begin{equation}
\label{eq:gen_cond2}
  M(x) \frac{\partial g(x)u}{\partial x} + \frac{\partial g(x)u}{\partial x}^T \!\! M(x) 
 + \sum_i \frac{\partial M(x)}{\partial x_i}[g(x)u]_i  =  0 
\end{equation}
\begin{equation}
\label{eq:gen_cond3}
 \frac{\partial h(x)}{\partial x}^T = M(x)g(x) \   . 
\end{equation}
In fact, along the solutions to the prolonged system, 
the time derivative of $\delta S$ is given by 
$\dot{\delta S} = \frac{1}{2}\delta x^T ( m_f(x) + m_g(x,u) )\delta x + \delta x^T M(x) g(x) \delta x$,
where $m_f(x)$ and $m_g(x,u)$ are given by the left-hand sides of
\eqref{eq:gen_cond1} and \eqref{eq:gen_cond2}, respectively. 

\eqref{eq:gen_cond1} guarantees that the system is 
contracting for $u=0$, thus incrementally stable with respect
to the geodesic distance induced by the metric $M$. The reader
will notice that \eqref{eq:gen_cond1} 
is just the usual condition for passivity
$\frac{\partial S(x)}{\partial x}f(x) \leq 0$ lifted to the tangent bundle. 
In a similar way, \eqref{eq:gen_cond3} guarantees that 
$\delta y = M(x) g(x) \delta x $, thus enforcing a 
differential version of the 
passivity condition $\frac{\partial S(x)}{\partial x}g(x) =  h(x)^T$.

A notable difference with respect to passivity is provided by condition
\eqref{eq:gen_cond2}, which requires 
the columns of $g(x)$ to be killing vector fields for the metric $M(x)$.
This guarantees that $u$ does not appear in the right-hand side of
$\dot{\delta S}$, as required by \eqref{eq:passivity}.
In this sense, the input matrix $g(x)$ restricts the class of metrics
that one can use to establish differential passivity.

For the case $g(x) = B$, for example, \eqref{eq:gen_cond2} restricts
the differential storage within the class of metrics $M(x)$ such that 
$\sum_i \frac{\partial M(x)}{\partial x_i}[Bu]_i = 0$, which is
satisfied by constant metrics $M(x) = P=P^T \geq 0$. 
In comparison to passivity, $M(x) = P$ is not an issue
for linear systems
\begin{equation}
\label{eq:lin}
\left\{
\begin{array}{rcl}
\dot{x} &=& Ax + Bu \\
y &=& C x  \\
\end{array}
\right.
\end{equation}
($A \in \real^{n\times n}$, $B \in \real^{n\times \nu}$, 
and $C \in \real^{\nu \times n}$). In fact, for
passive linear systems one  
can always find $P=P^T \geq 0$ such that
\begin{equation}
\label{eq:pc}
\begin{array}{rclcrcl}
A^T P + PA &\leq &0 &\quad&
 C^T &=& PB \ ,
\end{array}
\end{equation}
which also establishes 
the equivalence between passivity and differential passivity
for linear systems.
But $M(x) = P$ determines a limitation 
for the satisfaction of \eqref{eq:gen_cond1}
on systems of the form 
\begin{equation}
 \dot{x} = f(x) + Bu
\end{equation}
since it reduces 
\eqref{eq:gen_cond1} to {
$\frac{\partial f(x)}{\partial x} ^T P + P\frac{\partial f(x)}{\partial x} \leq 0$.
This last inequality
coincides with the 
early convergence condition of Demidovich
\cite{Demidovich1961}. See also \cite[Theorem 2.29]{Pavlov2005}.
It also resembles a classical Lyapunov inequality
based on quadratic Lyapunov functions and linearized vector fields. 
In fact, in the neighborhood of stable equilibria $x_e$ 
passivity and differential passivity are related, since 
locally around $x_e$ passive systems satisfies 
$\frac{\partial f(x)}{\partial x} ^T P + P\frac{\partial f(x)}{\partial x} \leq 0$
locally around $x_e$.}

The relevance of the condition enforced by \eqref{eq:gen_cond2}
is readily illustrated by the following example.
\begin{example}
\label{ex:osc}
Consider the simple dynamics on $\mathbb{S}$ given by
\begin{equation}
\label{eq:osc}
 \dot{x} = - \sin(x) + g(x)u \qquad g(x)=1 \ . 
\end{equation}
For $g(x) = 1$, 
\eqref{eq:gen_cond2} { allows for differential storages of the form}
$\delta S = \frac{1}{2}\delta x^2$, for which
$\dot{\delta S} = - \cos(x) \delta x^2 + \delta x \delta u$. 
Thus, \eqref{eq:osc} is differentially passive along solution curves
whose range belongs to $[-\frac{\pi}{2},\frac{\pi}{2}]$. In fact, \eqref{eq:gen_cond1} 
holds only for $x\in [-\frac{\pi}{2},\frac{\pi}{2}]$. 

Using a non constant metric,  \eqref{eq:gen_cond1} can be satisfied in 
the whole set $(-\pi,\pi)$. Indeed, taking
\begin{equation}
\delta S = M(x)\delta x^2 \qquad M(x) = \frac{1}{1+\cos(x)}
\end{equation}
\eqref{eq:gen_cond1} reads
\begin{equation}
- \frac{2\cos(x)}{1+\cos(x)} - \frac{\sin(x)^2}{(1+\cos(x))^2} = -1 \ .
\end{equation}
However, \eqref{eq:gen_cond2} does not hold, unless the input matrix $g(x) = 1$ in 
\eqref{eq:osc} 
is replaced by $g(x) = \gamma \cos(\frac{x}{2})$, where $\gamma \in \real$. 
In such a case, following \eqref{eq:gen_cond3},
\eqref{eq:osc} is differential passive with respect to the
output $y = \gamma\int_0^x \frac{\cos(\frac{z}{2})}{1+\cos(z)} dz$.
\end{example}

The discussion above makes clear that  
differential passivity for nonlinear systems of the form \eqref{eq:sys}
can be established only for suitable pairs $f(x)$ and $g(x)$.
The latter, through \eqref{eq:gen_cond2}, defines the class of feasible metrics. 
The former, through \eqref{eq:gen_cond1}, is required to be 
a contractive vector field with respect to a feasible metric
(see \cite{Forni2012arxiv,Lohmiller1998}). Finally, in analogy with passivity,
the (differential) passivating output depends on the differential storage and on the input
matrix, as established by \eqref{eq:gen_cond3}.

\section{Open gradient systems}
\label{sec:gradient}

\subsection{General formulation and prolonged system}
Given a smooth manifold $\mathcal{X}$, a Riemannian 
metric $Q$ on $\mathcal{X}$, and a potential function $V:\mathcal{X}\to\real$,
the local coordinates representation of a gradient system is given by
\begin{equation}
\label{eq:grad}
 Q(x) \dot{x} = - \frac{\partial V(x)}{\partial x} + B u \ .
\end{equation}
Following the discussion of the previous section, 
the study of differential passivity for gradient systems amounts
to verify that $f(x) := Q(x)^{-1}  \frac{\partial V(x)}{\partial x}$ and
$g(x) := Q(x)^{-1} B$ satisfy \eqref{eq:gen_cond1}, \eqref{eq:gen_cond2}
{ for some} differential storage $\delta S = \frac{1}{2}\delta x^TM(x)\delta x$.

The prolonged system is given by
\eqref{eq:grad} and by the variational system 
\begin{equation}
\label{eq:dgrad}
 Q(x)\dot{\delta x} = -\left[ \frac{\partial^2 V(x)}{\partial x^2}\delta x \!+\! B\delta u \right]  
 + \Gamma\!\left(x,u,\frac{\partial V(x)}{\partial x}\right) \!\delta x 
\end{equation}
where the matrix $\Gamma$ satisfies
\begin{equation}
\label{eq:Gamma}
\Gamma\!\left(x, u, \frac{\partial V(x)}{\partial x}\right) \!\delta x 
:= 
 -\left[ \sum_{i} \frac{\partial Q(x)}{\partial x_i} \delta x_i \right]\dot{x} \ . \nonumber \\
\end{equation}
$\Gamma$ is homogeneous of degree one in $\frac{\partial V(x)}{\partial x}$ and $u$, 
thus converges to 
zero as $x$ approaches an extremal point of $V$ and 
$u$ converges to $0$. 
{
Note that $\Gamma = 0$ when $Q(x)$ is constant}.

\subsection{Differential passivity via natural metric and convexity}
For $M(x) = Q(x) = P>0$ (constant), the differential storage 
$\delta S = \frac{1}{2} \delta x^T P \delta x$ 
guarantees that both \eqref{eq:gen_cond1} and \eqref{eq:gen_cond2} hold,
provided that $\frac{\partial^2 V}{\partial x^2} \geq 0$ for all $x \in \mathcal{X}$. 
In fact, along
the solutions of the prolonged system, we have 
\begin{equation}
 \dot{\delta S} = - \underbrace{\delta x^T \frac{\partial^2 V(x)}{\partial x^2} \delta x}_{\geq 0} + 
 \delta x^T B \delta u \ .
\end{equation}
Thus, the gradient system is differentialy passive with respect to the output
$y = B^T x$.  

The case of $Q(x)$ non constant is more involved.
For $M(x) = Q(x)$ conditions \eqref{eq:gen_cond1} and \eqref{eq:gen_cond2}
may not hold, in general. In fact, along the solutions of the prolonged
system, the differential storage 
$\delta S = \frac{1}{2} \delta x^T Q(x) \delta x$ 
has derivative
\begin{flalign}
\label{eq:grad_d_diffstorage}
 \dot{\delta S} 
 &= - \, \delta x^T \frac{\partial^2 V(x)}{\partial x^2} \delta x + \delta x^T B \delta u \nonumber\\
 & + \frac{1}{2}\delta x^T \Gamma\left(\!x,\!u,\!\frac{\partial V(x)}{\partial x}\!\right)^T\!\!\!\!\delta x 
 + \frac{1}{2} \delta x^T \Gamma\left(\!x,\!u,\!\frac{\partial V(x)}{\partial x}\!\right)\delta x \nonumber \\  
 &  + \frac{1}{2}\delta x^T\Omega\left(\!x,\!u,\!\frac{\partial V(x)}{\partial x}\!\right) \delta x 
\end{flalign}
where 
\begin{equation}
\Omega\left(x,u,\frac{\partial V(x)}{\partial x}\right) := 
 \sum_{i} \frac{\partial Q(x)}{\partial x_i} \dot{x}_i \ ;
\end{equation}
and \eqref{eq:gen_cond1} and \eqref{eq:gen_cond2} are
equivalent to the following inequality
\begin{equation}
\label{eq:cond5}
\delta x^T 
\left( 
-\frac{\partial^2 V(x)}{\partial x^2} + \Gamma^T + \Gamma + \Omega \right) 
\delta x \leq 0 \ . 
\end{equation} 
When \eqref{eq:cond5} holds 
for each $(x,\delta x) \in T\mathcal{X}$ and $u \in \mathcal{U}$, then
\eqref{eq:grad} is differentially passive with respect to the
output $y= B^T u$. 

\begin{example}[Example \ref{ex:osc} revised]
\label{ex:osc2}
Taking $V(x) = 1-\cos(x)$ and $g(x) = 1$ the dynamics in \eqref{eq:osc}
reads
\begin{equation}
\label{eq:osc2}
 \dot{x} = -\frac{\partial V(x))}{\partial x} + u \ .
\end{equation}
Note that $V(x)$ is convex in the region $[-\frac{\pi}{2},\frac{\pi}{2}]$
since $\frac{\partial^2 V(x)}{\partial x} = \cos(x) \geq 0$
for $x\in [-\frac{\pi}{2},\frac{\pi}{2}]$. In fact, \eqref{eq:osc2}
is differentially passive with 
$\delta S = \delta x^2$ and $y = x$. 

For $g(x) = \cos(\frac{x}{2})$, define 
$Q(x)=\frac{1}{\cos(\frac{x}{2})}$ and $V(x) = -4\cos(\frac{x}{2})$.
Then, \eqref{eq:osc} is well defined in $(-\pi, \pi)$ and reads
\begin{equation}
\label{eq:osc3}
 Q(x)\dot{x} 
 = -\frac{\sin(x)}{\cos\left(\frac{x}{2}\right)} + u 
 = - 2\sin\left(\frac{x}{2}\right) + u 
 = - \frac{\partial V(x)}{\partial x} + u\ .
\end{equation}
$V(x)$ is convex in $(-\pi,\pi)$, however differential dissipativity
cannot be achieved because the term 
$\Gamma^T+ \Gamma + \Omega$ in \eqref{eq:cond5} shows a dependence on $u$.
\end{example}

\begin{remark}
When \eqref{eq:cond5} does not hold, we can still achieve local differential passivity 
under the assumption of strict convexity of $V$, for small signals $u$.
Given a (sufficiently small) neighborhood $\mathcal{C}(x_e)$,
$\frac{\partial^2 V(x)}{\partial x^2}> a I$ for $x\in \mathcal{C}(x_e)$, 
while the last three terms in \eqref{eq:cond5} are bounded by
a function of the form $b(x_e)|u||\frac{\partial V(x)}{\partial x}|$,
by homogeneity.
Thus, 
$\dot{\delta S} \leq \left(- a + b(x_e) |u|\left|\frac{\partial V(x)}{\partial x}\right|\right)
 |\delta x|^2 + \delta x^T B \delta u \leq \delta x^T B \delta u $
for $x\in C(x_e)$ and for $|u|$ and $C(x_e)$ sufficiently small. 
\end{remark}

\subsection{Differential passivity beyond the natural metric}
We consider the case of differential storage functions 
$\delta S = \frac{1}{2}\delta x^T M(x) \delta x$ 
{ where 
$M(x) = Q(x)PQ(x)$ for
some given matrix $P=P^T\geq  0$}. 
A first consequence of the definition of $M(x)$ is that $Q$ can be relaxed to a pseudo-Riemannian
metrics, that is, $Q(x)$ is not necessarily positive but still invertible.  
In contrast to this generalization effort, we restrict $Q$ to the
class of pseudo-metrics defined by $Q(x) = \frac{\partial^2 q(x)}{\partial x^2}$,
where $q$ is a function differentiable sufficiently many times.

Under these assumptions, for
\begin{equation}
y = C \frac{\partial q(x)}{\partial x}
\end{equation}
\eqref{eq:gen_cond1}, \eqref{eq:gen_cond2},
and \eqref{eq:gen_cond3} are { equivalent to
the following conditions.}
\begin{theorem}
\label{thm:main_result}
Consider $q:\mathcal{X} \to \real$ and $Q(x) = \frac{\partial^2 q}{\partial x^2}(x)$. Then (18) is differentially passive with respect to the output $y = C\frac{\partial q(x)}{\partial x}$ if there exists a matrix $P=P^T \geq 0$ such that for all $x \in \mathcal{X}$ 
\begin{subequations}
\label{eq:cond_linear-like}
\begin{eqnarray}
 \frac{\partial^2 V(x)}{\partial x^2} PQ(x) + 
 Q(x)P\frac{\partial^2 V(x)}{\partial x^2}  
 &\geq &0 \label{eq:cond_linear-like1}\\
 C^T&=&PB \, . \label{eq:cond_linear-like2}
 \end{eqnarray}
\end{subequations}
{
$\delta  S = \frac{1}{2} \delta x^T Q(x)PQ(x) \delta x$
is the differential storage
}
\end{theorem}

\eqref{eq:cond_linear-like1} 
is a generalized convexity property on $V$. 
We get classical convexity when $Q(x)=P=I$.
For $P$ positive definite, 
the particular selection of the output $y = C \frac{\partial q(x)}{\partial x} $
guarantees that \eqref{eq:grad} has relative degree one. In fact,
$\dot{y} = C \frac{\partial^2 q(x)}{\partial x^2} \dot{x} = C (\frac{\partial V(x)}{\partial x}(x) + Bu)$, where
$CB = B^TPB$. Finally, note that 
for $q(x) = V(x)$, the inequality in \eqref{eq:cond_linear-like} is always satisfied. 
This is not surprising since, by defining 
$e = \frac{\partial V(x)}{\partial x}$,  \eqref{eq:grad} reads
$\dot{e} = - e + B u$, $y = Ce$.

\begin{proofof}\emph{Theorem \ref{thm:main_result}:}
Define 
$f(x) := \left[\frac{\partial^2 q(x)}{\partial x^2}\right]^{-1} \frac{\partial V(x)}{\partial x}$, 
$g(x) := \left[\frac{\partial^2 q(x)}{\partial x^2}\right]^{-1}Bu$, and
$h(x) := C^T \frac{\partial q(x)}{\partial x}$,
and consider the prolonged system \eqref{eq:sys},\eqref{eq:dsys}.
By exploiting the differentiability of $q$, and using the chain rule, 
\begin{equation}
\label{eq:dW_IO}
\begin{array}{rcl}
 \dot{\delta S} 
 \!\!&\!\!=\!\!&\!\!\! \delta x^T Q(x) P \!\left[
 \dfrac{\partial [Q(x)f(x)]}{\partial x}  \delta x 
 + \dfrac{\partial [Q(x)g(x)u]}{\partial x}  \right]
 \delta x \vspace{1mm}\\
 & & + \ \delta x^T Q(x) PQ(x)g(x) \delta u  \vspace{2mm} \\
&=& \underbrace{\delta x^T Q(x) P
 \dfrac{\partial V(x)}{\partial x}  \delta x}_{\leq 0} 
 + \ \delta x^T Q(x) P
 \underbrace{\dfrac{\partial [Bu]}{\partial x} }_{=0} \delta x \\
 &&\ +\  \delta x^T Q(x) \underbrace{P B}_{C^T} \delta u  \vspace{2mm} \\
&\leq&
 \delta x^T Q(x) C^T\delta u \ = \ \delta y^T \delta u \ . 
\end{array} 
\end{equation}
\eqref{eq:gen_cond1}, \eqref{eq:gen_cond2}, \eqref{eq:gen_cond3}
read
$\delta x^T Q(x) P \frac{\partial V(x)}{\partial x} \delta x \leq 0$,
$\delta x^T Q(x) P \frac{\partial [Bu]}{\partial x} \delta x = 0$, 
and $\delta x^T Q(x) P B  = \delta x Q(x) C^T$, respectively.
\end{proofof}
\begin{remark}
Theorem \ref{thm:main_result} extends to systems
of the form 
\begin{equation}
\label{eq:gen}
 Q(x) \dot{x} = A(x) + B u 
\end{equation}
where $A(x)$ is a vector field not derived from a potential.
In this case, $\frac{\partial^2 V(x)}{\partial x^2}$ in 
\eqref{eq:cond_linear-like1} is replaced by $\frac{\partial A(x)}{\partial x}$.
\end{remark}
\begin{example}[Example \ref{ex:osc2} revised]
Consider the system formulation given in \eqref{eq:osc3} 
for the case $g(x)=\cos\left(\frac{x}{2}\right)$. Take the differential storage
$\delta S = \frac{1}{2} \delta x^T Q(x)PQ(x)\delta x$ for $P = 1$.
Then, from Theorem \ref{thm:main_result}, 
the inequality \eqref{eq:cond_linear-like1} reads 
\begin{equation}
\label{eq:osc_result}
 2\frac{\cos\left(\frac{x}{2}\right)}{\cos\left(\frac{x}{2}\right)} = 2 \geq 0 \ ,
\end{equation}
and \eqref{eq:osc3} is differentially passive in $(-\pi,\pi)$ with
respect to the output $y = \int_0^x Q(z) dz $.
{ Because \eqref{eq:osc_result} is strictly positive, the system
is incrementally asymptotically stable.
The solutions converge to the unique steady-state 
solution compatible with the input signal $u$ \cite{Forni2012arxiv}
(see Fig \ref{fig:osc_results})}.
\begin{figure}[htbp]
\centering
\includegraphics[width=0.49\columnwidth]{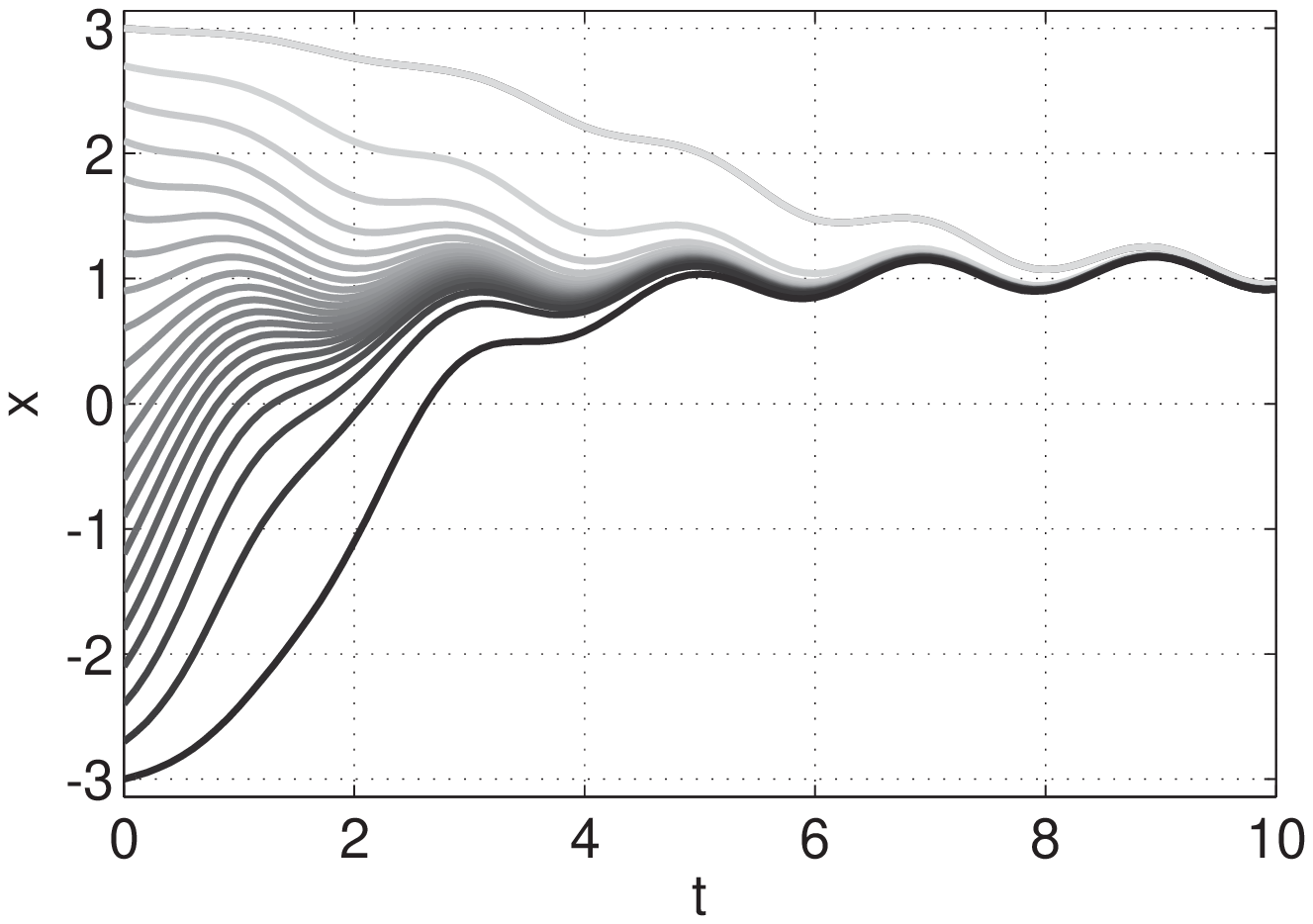} 
\includegraphics[width=0.49\columnwidth]{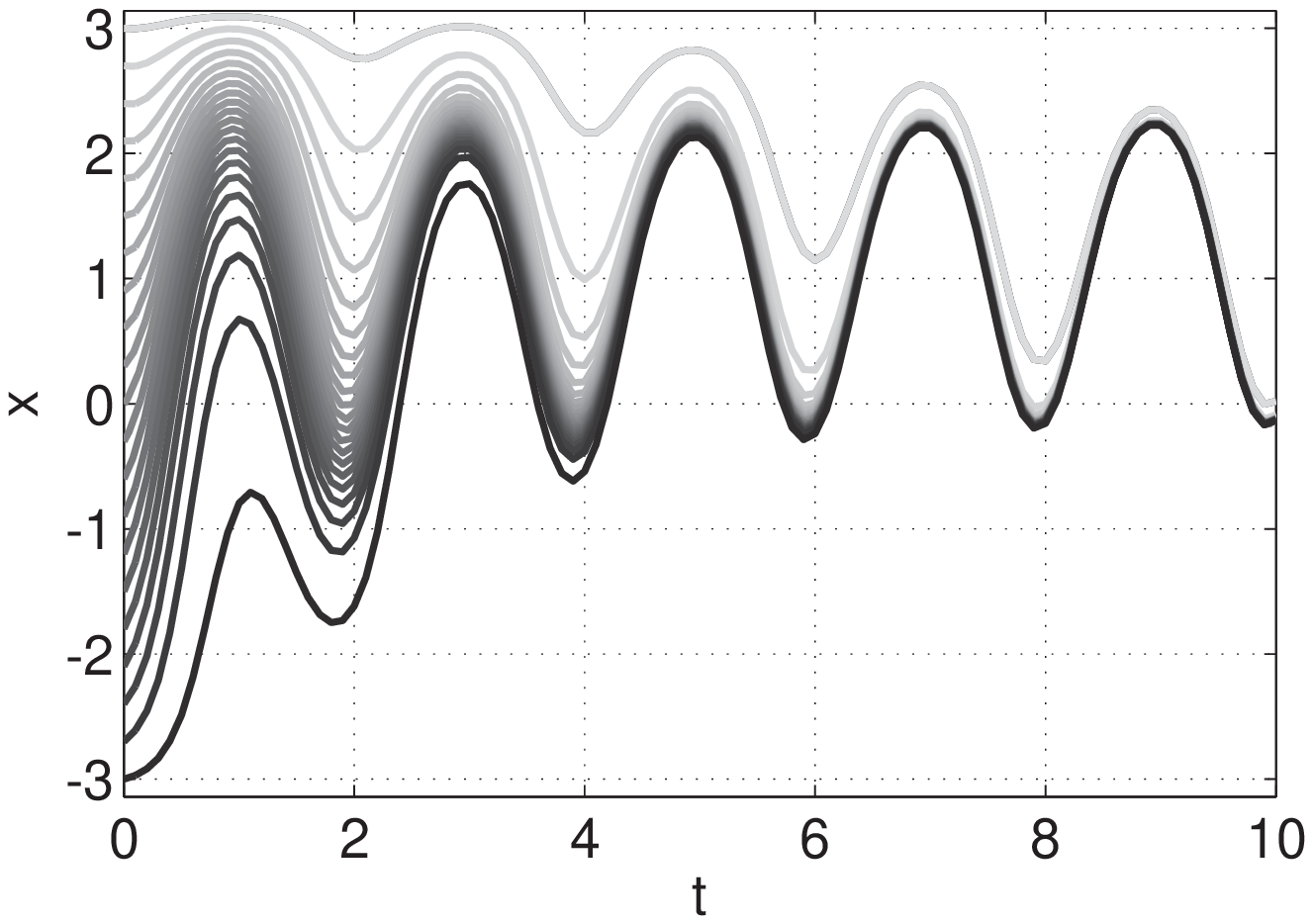} 
\vspace{-7mm}
\caption{{ Entrainment of \eqref{eq:osc3}
 with $g(x) = \cos\left(\frac{x}{2}\right)$
for the (small) input $u=1+0.5\sin(\pi t)$, left, and the (large)
input $u=1+5\sin(\pi t)$, right.}}
\label{fig:osc_results}
\end{figure}
\end{example}

\section{Brayton-Moser systems}
\label{sec:brayton_moser}

\subsection{Passivity conditions}
The approach developed in the previous section allows for the 
analysis of the passivity of Brayton-Moser systems 
\cite{Jeltsema2003,Jeltsema2009,Schaft2011}. 
Brayton-Moser modeling of physical systems 
characterizes a class of gradient systems  of the form 
\begin{equation}
\label{eq:brayton}
Q(z)\dot{z} = \frac{\partial V(z,u)}{\partial z} \ ,
\end{equation}
{
where the state-s[ace is given by flow and efforts $z = (f,e)$, $V$ is a the potential, 
and $Q(z)$ satisfies}
\begin{equation}
\label{eq:Q}
Q(z) = \mymatrix{cc}{- \frac{\partial^2 H^*(f,e)}{\partial f^2} & 0 \\
0 & \frac{\partial^2 H^*(f,e)}{\partial e^2}} \ .
\end{equation}
$H^*$ is the Legendre transform of the Hamiltonian $H$.
In relation to the theory developed in the previous section,
we assume that $H^*$ has the following structure
\begin{equation}
\label{eq:H*}
 H^*(f,e) = H^*_f(f)+H^*_e(e)
\end{equation}
which guarantees that $Q(z) = \frac{\partial^2 [-H^*_f(f)+H^*_e(e)]}{\partial z^2}$.
In a similar way, we assume that $V$ has the form 
\begin{equation} 
V(z,u) = p(z)+ z^TBu \ .
\end{equation}
{
Under these assumptions}, \eqref{eq:brayton} reads
\begin{equation}
\label{eq:brayton2}
\frac{\partial^2 H^*(z)}{\partial z^2} \dot{z} 
= \frac{\partial p(z)}{\partial z} + Bu \ .
\end{equation}
From Theorem \ref{thm:main_result},
the system \eqref{eq:brayton2} 
is differential passive
with respect to the output $y = B^T \frac{\partial H^*(z)}{\partial z}$, if
\begin{equation}
\frac{\partial^2 H^*(z)}{\partial z^2} \frac{\partial^2 p(z)}{\partial z^2}
+ \frac{\partial^2 p(z)}{\partial z^2} \frac{\partial^2 H^*(z)}{\partial z^2} 
\leq 0 \ .
\end{equation}
The reader will notice that the 
output $y =  B^T \frac{\partial H^*(z)}{\partial z}$ \emph{is not} the usual
passive output $y_p = B^T z$. 
However, $y$ and $y_p$ show an intriguing duality, through
energy and co-energy formulation of the system \cite[Section 4]{Schaft2011}.

\subsection{Differential passivity of a nonlinear RC circuit}

The behavior of the nonlinear circuit represented in Figure \ref{fig:nonlin_RC}
is captured by the following equations:
\begin{equation*}
\dot{q} = -i_r + i, \ \\ i_r = R(v), \  \ v= \frac{\partial h}{ \partial q}( q), \,q = C(v) =\frac{\partial h^*}{\partial v}(v) \ .
\end{equation*} \vspace{-4mm}
\begin{figure}[htbp]
\centering
\includegraphics[width=0.3\columnwidth]{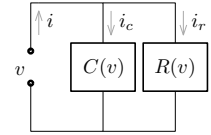}  \vspace{-2mm}
\caption{$V$,$I$ - external voltage and current. $v_c$,$i_c$ 
- capacitor voltage and current. $v_r$,$i_r$ - resistor voltage and current. }
\label{fig:nonlin_RC}
\end{figure}

\noindent
Defining $Q(v) = \frac{d^2 h^*}{d v^2}(v)$, we get the gradient system
\begin{equation}
\label{eq:RC-circuit}
Q(v) \dot{v} = - R(v) + i \ .
\end{equation}
From Theorem \ref{thm:main_result}, differential passivity can be achieved if 
$Q(v)\frac{\partial R(v)}{\partial v} \geq 0$. In fact, 
defining  $\delta S(v, \delta v) = \frac{1}{2} (Q(v)\delta v)^2$, we have that
\begin{equation}
\dot{S} = - Q(v) \frac{\partial R(v)}{\partial v}\delta v^2 + Q(v) \delta v \delta i \ .
\end{equation}
Therefore, if 
$R(v) $ is not decreasing and 
$ \frac{\partial h^*(v)}{\partial v}$ is strictly increasing, 
we get 
\begin{equation}
\dot{S} \ \leq \ Q(v)\delta v \delta i \ = \ \delta q \delta i \ .
\end{equation}
For example, suppose that $v$ can only take positive values, and take
$R(v) = v^5$. $R(v)$ models a nonlinear resistor $v= \tilde{R}(i)i$ whose value
$\tilde{R}(i)$ decreases as $i$ increases.  
For the capacitor, consider the relation $C(v) =  \frac{\partial h^*}{\partial v}(v) = \log(1+v)$, to
model a saturation effect on the capacitor plates, where the charge on the plates 
grows at sub-linear rate with respect to the voltage. Note that
$Q(v) = \frac{1}{1+v}>0$ for $v\geq 0$.

The incremental stability property of the circuit is clearly
visible in the left part of Figure \ref{fig:nonlin_RC_results1}. 
The steady-state behavior of the circuit is 
independent from the initial condition, { 
(nonlinear filter).}
\begin{figure}[t]
\centering
\includegraphics[width=0.49\columnwidth]{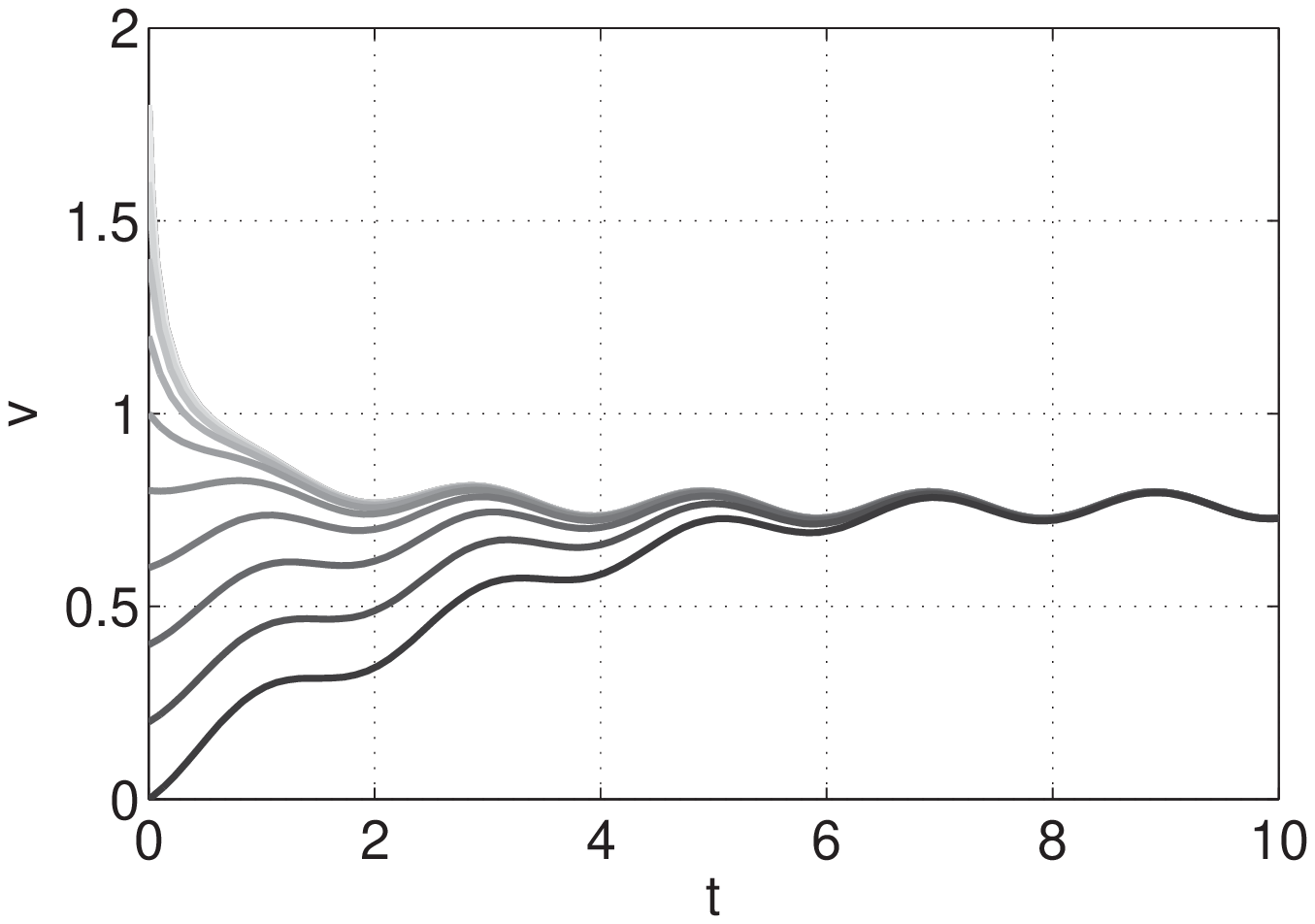} 
\includegraphics[width=0.49\columnwidth]{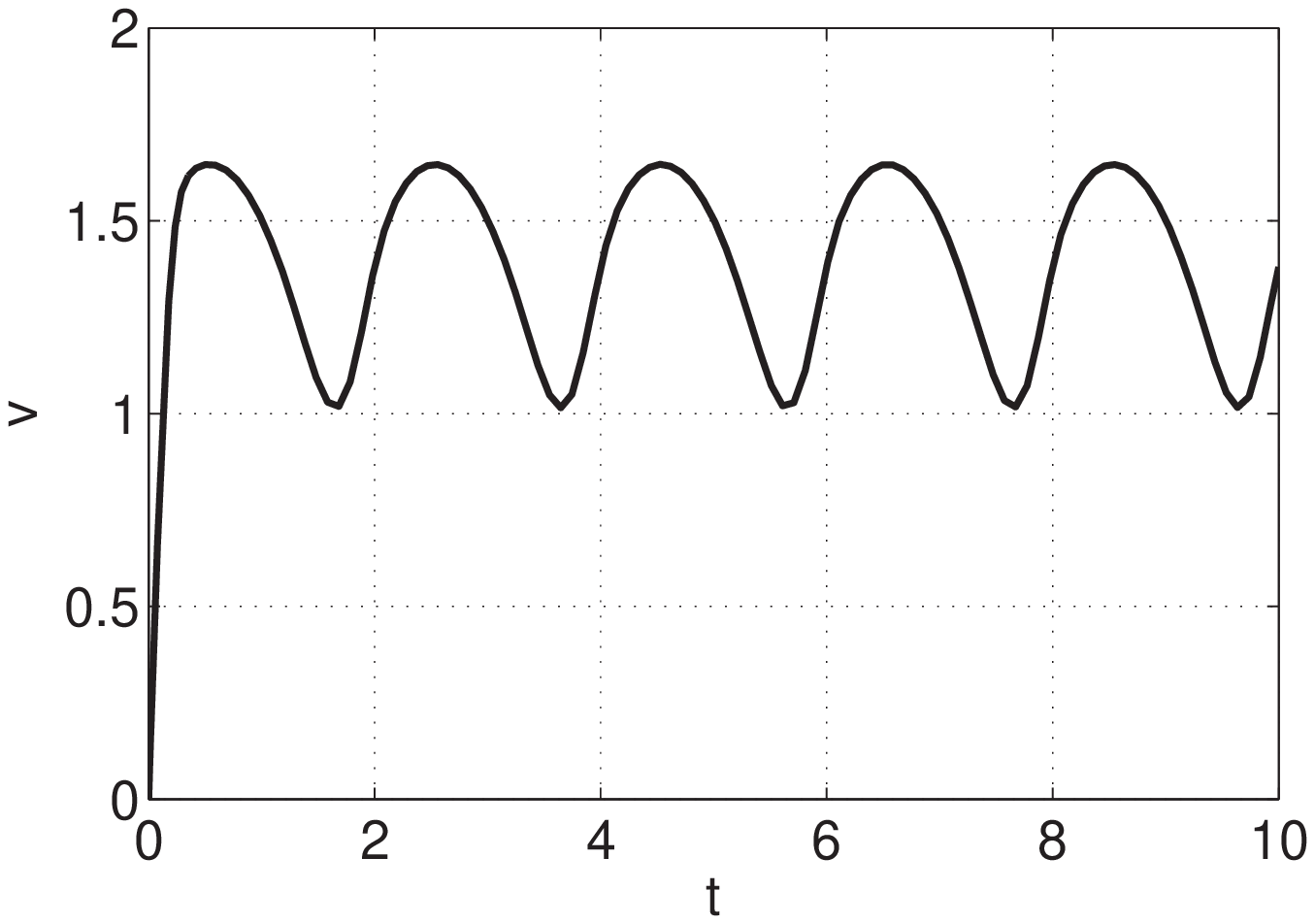} 
\vspace{-6mm}   
\caption{Contraction and nonlinear behavior of the nonlinear RC circuit. 
{
The left figure illustrates contraction for a broad range of initial conditions.
The right figure illustrates the nonlinear response of the circuit to 
a large harmonic input signal.}}

\label{fig:nonlin_RC_results1}
\end{figure}

\subsection{Differential passivation of the rigid body}
Let us consider the rigid-body dynamics given by
\begin{equation}
\begin{array}{rcl}
\smallmat{I_1 & 0 & 0\\ 0& I_2 & 0\\ 0 & 0 & I_3}
\dot{w} \!\!&\!\!=\!\!&\!\!
\smallmat{I_2-I_3 & 0 & 0 \\ 0 & I_3-I_1 & 0 \\ 0 & 0 & I_1-I_2} \!\!
\mymatrix{c}{\!\omega_2\omega_3\! \\ \!\omega_1\omega_3\! \\ \!\omega_1\omega_2\!}+u 
\end{array}
\end{equation} 
where $\omega_k$ and $I_k$ are the angular velocities of the body with respect to the axis of a frame 
fixed to the body, and the principle moments of inertia.

Suppose that $I_1>I_2>I_3$ and define
\begin{equation}
\begin{array}{rcl}
I &:=& \mathrm{diag}(I_1,I_2,I_3) \\
\tilde{Q} &:= & \mathrm{diag}(I_2-I_3 , I_3-I_1 ,  I_1-I_2) \\
Q &:=&  I \tilde{Q}^{-1} \\
p(\omega) &:=& \omega_1\omega_2\omega_3  \\
\end{array}
\end{equation}
then we can rewrite the rigid body dynamics as follows
\begin{equation}
\label{eq:rb}
Q \dot{\omega} = \frac{\partial p(\omega)}{\partial \omega} + \tilde{Q}^{-1} u
\qquad ( q(\omega) = \frac{1}{2}\frac{\partial^2 \omega^T Q \omega}{\partial \omega^2} ) \ .
\end{equation}
Furthermore, let us consider a passivation design given by
\begin{equation}
u =  I( -r(\omega)  +   Gv) \ , \quad 
 r(\omega) := 
 \mymatrix{ccc}{r_1\omega_1&
 r_2 \omega_2 &
 r_3\omega_3}^T.
\end{equation}
\eqref{eq:rb} becomes
\begin{equation}
Q \dot{\omega} = \frac{\partial p(\omega)}{\partial \omega} -Qr (\omega) + Q Gv \ .
\end{equation}
From Theorem \ref{thm:main_result}, picking $P=Q^{-2}$, 
\eqref{eq:cond_linear-like1} reads
\begin{equation}
\label{eq:rb_cond1}
Q^{-1}\frac{\partial^2 p(\omega)}{\partial \omega^2} +
\frac{\partial^2 p(\omega)}{\partial \omega^2} Q^{-1}
- 2\frac{\partial r(\omega)}{\partial \omega} 
\leq 0
\end{equation}
while condition \eqref{eq:cond_linear-like1} becomes
$C^T = Q^{-1}  G$.
Therefore, differential passivity from $v$ to $y= G^T \omega$
can be guaranteed semi-globally, since for any given
compact region of velocities, there exists a selection of $r_1,r_2,r_3$ 
that guarantees \eqref{eq:rb_cond1} within that region.

For $I_1 = 3$, $I_2 = 2$, $I_3=1$ and $r_1=r_2=r_3=0.2$,
to achieve a desired steady-state solution $[d(t),0,0]^T$ 
it is sufficient to define $G = [1,0,0]^T$ and 
$v = r_1 d(t) + \dot{d}(t) $,
{ as shown on the left of 
Figure \ref{fig:rb_results} for $d(t) = 3 \sin(\pi t)$.}
Using differential passivity, we can improve the convergence rate
by output feedback $v = -0.5 y + (r_1+0.5) d(t) + \dot{d}(t)$, as shown in the
simulation on the right.

\begin{figure}[htbp]
\centering
\hspace{-3mm}
\includegraphics[width=0.5\columnwidth]{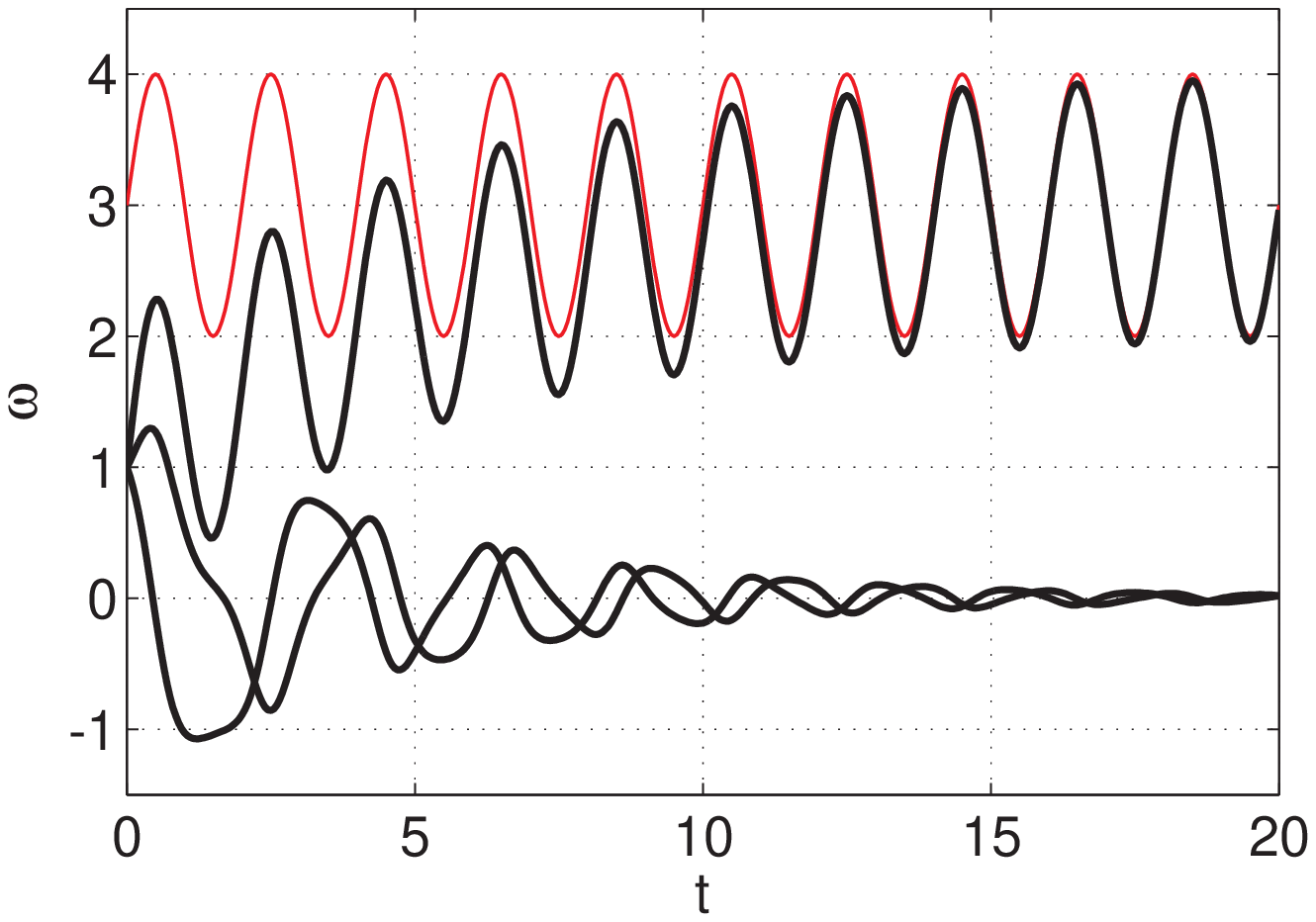}  
\includegraphics[width=0.5\columnwidth]{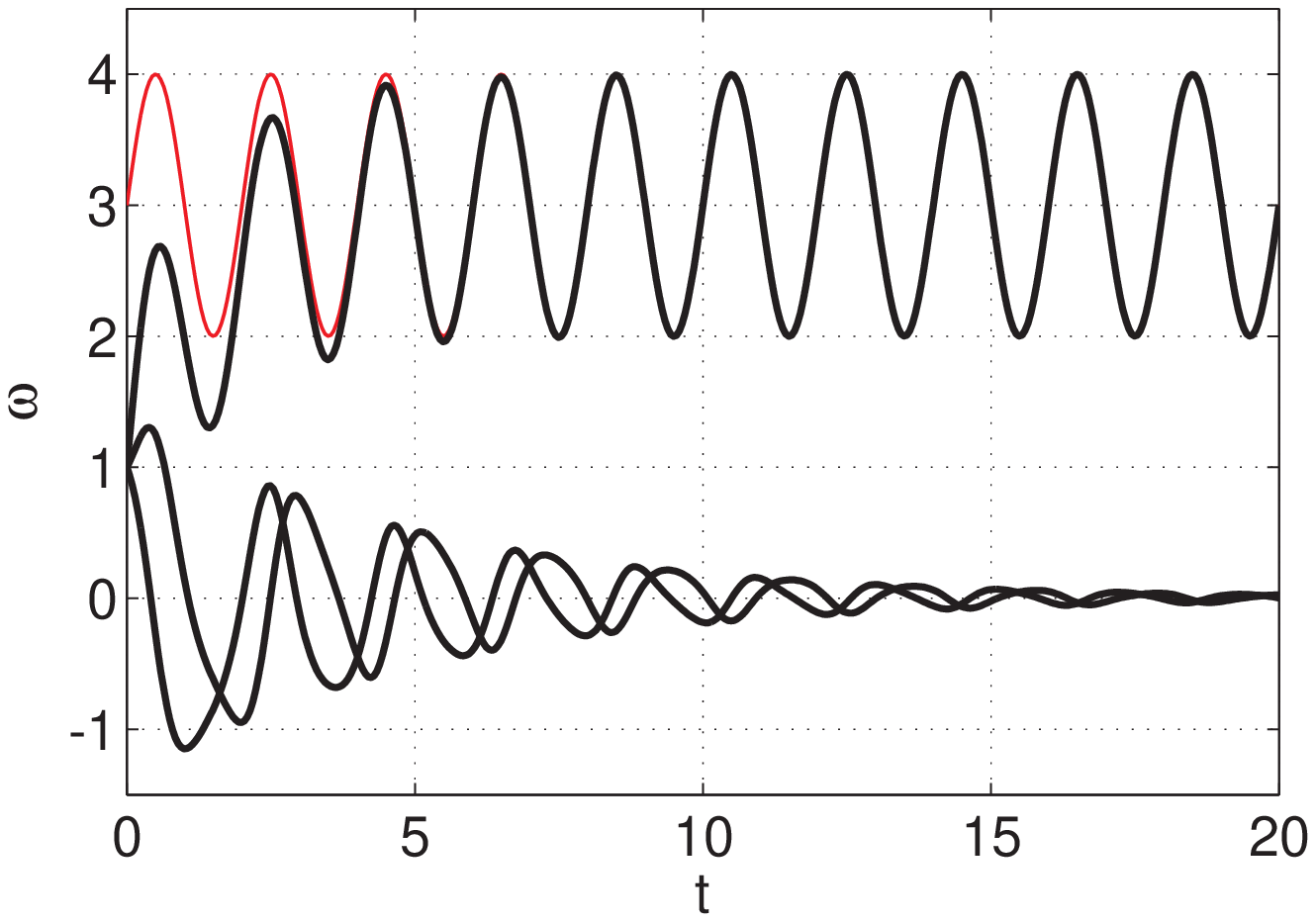}  
\vspace{-6mm}
\caption{The passivation design on the rigid body guarantees
contraction. The left figure illustrates the contraction of the three states.
Output injection $y = G^T \omega$
improves the convergence rate, as illustrated by the right figure.}
\label{fig:rb_results}
\end{figure}

\section{Conclusions}
Building upon \cite{Forni2013} and \cite{Schaft2013},
we introduced the notion of differential passivity and we
proposed geometric conditions for differential passivity of gradient 
and Brayton-Moser systems. The meaning and the feasibility
of such conditions is investigated through detailed
discussion and several examples. Examples suggests that
differential passivity may hold for a sizeable class of physical 
models.

\linespread{.94}
 \bibliographystyle{plain}

\end{document}